\documentclass[a4paper,UKenglish]{lipics}

\usepackage[utf8]{inputenc}

\usepackage[noend,vlined,figure]{algorithm2e} 
\usepackage[utf8]{inputenc}
\usepackage{color}
\usepackage{amsfonts}
\usepackage{amssymb}
\usepackage{amsmath}
\usepackage{multicol}
\usepackage{tikz}
\usepackage{prettyref}
\usepackage{pgfplots}
\usepackage{microtype}
\usepackage{enumerate}
\usepackage{listings}

\definecolor{blue}{rgb}{0.211,0.211,0.856}
\definecolor{red}{rgb}{0.856,0.2,0.2}
\newcommand{\blue}{\color{blue}}
\newcommand{\red}{\color{red}}

\definecolor{darkgreen}{rgb}{0.1,0.656,0.1}
\newcommand{\darkgreen}{\color{darkgreen}}

\definecolor{green}{rgb}{0.1,0.656,0.1}

\definecolor{lgreen}{rgb}{0.411,1.0,0.411}
\definecolor{lred}{rgb}{1.0,0.711,0.501}

\definecolor{lgrey}{rgb}{0.5,0.5,0.5}
\definecolor{llgrey}{rgb}{0.68,0.68,0.68}

\newenvironment{SE}{\noindent\color{red} SE : }{}

\newenvironment{aw}{\noindent\color{magenta} AW :  }{}

\newrefformat{thm}{Theorem~\ref{#1}}
\newrefformat{lem}{Lemma~\ref{#1}}
\newrefformat{dfn}{Definition~\ref{#1}}
\newrefformat{cor}{Corollary~\ref{#1}} 
\newrefformat{prop}{Proposition~\ref{#1}}
\newrefformat{sec}{Section~\ref{#1}}
\newrefformat{kap}{Chapter~\ref{#1}}
\newrefformat{bsp}{Example~\ref{#1}}
\newrefformat{rem}{Remark~\ref{#1}}
\newrefformat{fig}{Figure~\ref{#1}}
\newrefformat{eq}{(\ref{#1})}
\newrefformat{ex}{Example~\ref{#1}}
\newrefformat{tab}{Table~\ref{#1}}

\newcommand{\QuickMergeXsort}{{QuickMergeXsort}\xspace}
\newcommand{\QuickMergesort}{{QuickMergesort}\xspace}

\newcommand{\QuickXsort}{{QuickXsort}\xspace}

\newcommand{\Oh}{\mathcal{O}}

\newcommand{\floor}[1]{\left\lfloor\mathinner{#1} \right\rfloor}
\newcommand{\ceil}[1]{\left\lceil\mathinner{#1} \right\rceil}

\renewcommand{\Pr}[1]{\mathop{\mathrm{Pr}}\left[\,#1\,\right]}

  \title{QuickMergesort: Practically Efficient Constant-Factor Optimal Sorting}     
\titlerunning{QuickMergesort}

     \author[1]{Stefan Edelkamp}
     \author[2]{ Armin Wei\ss}
     \affil[1]{King's College London, UK\\
     	\texttt{stefan.edelkamp@kcl.ac.uk}}
\affil[2]{ Universit{\"a}t Stuttgart, Germany\\
	\texttt{armin.weiss@fmi.uni-stuttgart.de}}
     \authorrunning{S. Edelkamp and A. Wei\ss} 
     
     \Copyright{Stefan Edelkamp and Armin Wei\ss}
     
     \subjclass{F.2.2 Nonnumerical Algorithms and Problems}
     
     \serieslogo{}
     \volumeinfo
     {Billy Editor and Bill Editors}
     {2}
     {Conference title on which this volume is based on}
     {1}
     {1}
     {1}
     \EventShortName{}
     \DOI{10.4230/LIPIcs.xxx.yyy.p}

\begin{document}
     
\date{}

\maketitle
\thispagestyle{empty}

\definecolor{mygreen}{rgb}{0,0.6,0}
\definecolor{mygray}{rgb}{0.5,0.5,0.5}
\definecolor{mymauve}{rgb}{0.58,0,0.82}

\begin{abstract}
 We consider the fundamental problem of internally sorting a sequence of $n$ elements.  In its best theoretical setting QuickMergesort,
 a combination Quicksort with Mergesort with a Median-of-$\sqrt{n}$ pivot selection, requires at most $n \log n - 1.3999n + o(n)$ element comparisons on the average. The questions addressed in this paper is how to make this algorithm practical. As refined pivot selection usually adds much overhead, we show that the Median-of-3 pivot selection of QuickMergesort leads to at most $n \log n - 0{.}75n + o(n)$ element comparisons on average, while running fast on elementary data. 
 The experiments show that QuickMergesort outperforms state-of-the-art library implementations, including {C++}'s Introsort and Java's Dual-Pivot Quicksort. Further trade-offs between a low running time and a low number of comparisons are studied. Moreover, we describe a practically efficient version with $n \log n + \Oh(n)$ comparisons in the worst case.

\keywords{in-place sorting, quicksort, mergesort, analysis of algorithms}
\end{abstract}

\pagestyle{plain}
\section{Introduction}

Sorting a sequence of $n$ elements remains one of the most fascinating
topics in computer science, and runtime improvements to sorting has
significant impact for many applications.  The lower bound is $\log
(n!) \approx n \log n - 1.44n + \Theta(\log n)$ element comparisons
applies to the worst and the average case\footnote{Logarithms denoted by
$\log$ are base 2, and the term {average case} refers to a
  uniform distribution of all input permutations assuming all elements
  are different.}.

The sorting algorithms we propose in this paper are \emph{internal} or
\emph{in-place}: they need at most $\Oh(\log n)$ space (computer
words) in addition to the array to be sorted. That means we consider
{Quicksort}~\cite{Hoa62} an internal algorithm, whereas standard
{Mergesort} is \emph{external} because it needs a linear amount of
extra space.

Based on {QuickHeapsort} \cite{quickheap,DiekertW13Quick}, Edelkamp
and Wei\ss~\cite{EWCSR} developed the concept of {QuickXsort} and
applied it to X = {WeakHeapsort}~\cite{Dut93} and X = {Mergesort}.
The idea~-- going back to {UltimateHeapsort}~\cite{ultimate}~-- is
very simple: as in {Quicksort} the array is partitioned into the
elements greater and less than some pivot element, respectively. Then
one part of the array is sorted by X and the other part is sorted
recursively. The advantage is that, if X is an external algorithm,
then in {QuickXsort} the part of the array which is not currently
being sorted may be used as temporary space, which yields an internal
variant of X.

Using {Mergesort} as X, a partitioning scheme with 
$\lceil \sqrt{n} \rceil$ pivots, known to be
optimal for classical {Quicksort}~\cite{MartinezR01}, and Ford and
Johnson's {MergeInsertion} as the base case~\cite{FordJ59},
{QuickMergesort} requires at most $n \log n - 1.3999n + o(n)$ element
comparisons on the average (and $n \log n - 1.4n + o(n)$ for $n=2^k$),
while preserving worst-case bounds $n \log n + \Oh(n)$ element
comparisons and $\Oh(n \log n)$ time for all other
operations~\cite{EWCSR}. To the authors' knowledge the average-case
result is the best-known upper bound for sequential sorting with
$\Oh(n \log n)$ overall time bound, in the leading term matching, and
in the linear term being less than $0.045n$ away from the lower bound.
The research question addressed in this paper, is whether 
{QuickMergesort} can be made \emph{practical} in relation to
efficient library implementations for sorting, such as {Introsort}
and {Dual-Pivot Quicksort}. 

{Introsort}~\cite{Mus97}, implemented as \texttt{std::sort} in C++/STL, 
is a mix of {Insertionsort},
{CleverQuicksort} (the Median-of-3 variant of {Quicksort})
and {Heapsort}~\cite{Flo64,Wil64}, where the former and latter are
used as recursion stoppers (the one for improving the performance
for small sets of data, the other one
for improving worst-case performance). The average-time complexity,
however, is dominated by {CleverQuicksort}.

{Dual-Pivot Quicksort}\footnote{Oracle states:
\emph{The sorting algorithm is a Dual-Pivot Quicksort by Vladimir 
Yaroslavskiy, Jon Bentley, and Joshua Bloch}; see
  http://permalink.gmane.org/gmane.comp.java.openjdk.core-libs.devel/2628}
by Yaroslavskiy et al. as implemented in current versions of
Java (e.g., Oracle Java 7 and Java 8) is an interesting 
{Quicksort} variant using two (instead of one) pivot elements in the
partitioning stage (recent proposals
use three and more pivots~\cite{KushagraLQM14}).  It has been shown that
-- in contrast to ordinary {Quicksort} with an average case of
$2\cdot n \ln n +\Oh(n)$ element comparisons -- {Dual-Pivot Quicksort}
requires at most $1.9 \cdot n \ln n +\Oh(n)$ element comparisons on the
average, and there are variants that give $1.8
\cdot n \ln n +\Oh(n)$. For a rising number of samples for 
pivot selection, the leading factor 
decreases~\cite{dqsanalysis1,dqsanalysis2,dqsanalysis3}.

So far there is no \emph{practical} (competitive in
  performance to state-of-the-art library implementations)
sorting algorithm that is \emph{internal} and
\emph{constant-factor-optimal} (optimal in the leading term).
Maybe closest is InSituMergesort~\cite{KatajainenPT96,ElmasryKS12}, but even
though that algorithm improves greatly over the library
implementation of in-place stable sort in STL, it could not match with
other internal sorting algorithms.
Hence, the aim of the paper is to design fast
{QuickMergesort} variants. 
Instead of using a Median-of-$\sqrt{n}$ strategy, we will
use the Median-of-3. For Quicksort, the Median-of-3 strategy is also known as {CleverQuicksort}. 
The leading constant in $c \cdot n \log n +\Oh(n)$ for the average
case of comparisons of {CleverQuicksort} is $c = (12/7)\ln 2
\approx 1.188$. 
As $c < 1.8 \ln 2$, 
{CleverQuicksort} is theoretically superior to the wider class of
        {DualPivotQuicksort} algorithms considered
        in~\cite{dqsanalysis1,dqsanalysis3,dqsanalysis2}.

Another sorting algorithm studied in this paper is a mix of 
{QuickMergesort} and {CleverQuicksort}:  
during the sorting with {Mergesort}, for small arrays {CleverQuicksort} is applied.

The contributions of the paper are as follows.
\begin{enumerate}
	\item We derive a bound on the average number of comparisons in {QuickMergesort}
	when the Median-of-3 partitioning strategy is used instead of the
	Median-of-$\sqrt{n}$ strategy, and show a surprisingly low upper bound of
	$n \log n - 0.75n + o(n)$ comparisons on average. 
	\item We analyze a variant of {QuickMergesort} where base
          cases of size at most $n^\beta$ for some $\beta \in [0,1]$
          are sorted using yet another sorting algorithm X; otherwise the
          algorithm is identical to {QuickMergesort}. We show that
          if X is called for about $\sqrt{n}$ elements and X uses at most $\alpha \cdot n \log n+
          \Oh(n)$ comparisons on average, the average
          number of comparisons of is $(1+\alpha)/2 \cdot n \log n+ \Oh(n)$, with
          $(1+\alpha)/2 \approx 1.094$ for X $=$ Median-of-3
             {Quicksort}. Other element size thresholds for invoking X lead to
             further trade-offs.
    \item We refine a trick suggested in \cite{EWCSR} in order to obtain a bound of  $n \log n + 16.1n$ comparisons in the worst case using the median-of-median algorithm \cite{BFPRT73} with an adaptive pivot sampling strategy.
    On average the modified algorithm is only slightly slower than the Median-of-3 variant of QuickMergesort. 
	\item We compare the proposals empirically to other algorithms from
	the literature.  
\end{enumerate}
We start with revisiting {QuickXsort}
and especially {QuickMergesort}, including theoretically important
and practically relevant sub-cases. We derive an upper bound on
the average number of comparisons in {QuickMergesort} with
Median-of-3 pivot selection. In \prettyref {sec:QMQS}, we present changes to the
algorithm that lead to the hybrid {QuickMergeXsort}. Next, we introduce the worst-case efficient variant MoMQuickMergesort, and, finally, we present experimental results.

\section{ \QuickXsort\ and \QuickMergesort }\label{sec:quickXsort}

In this section we give a brief description of {QuickXsort} and
extend a result concerning the number of comparisons performed in the
average case.  

Let X be some sorting algorithm.  {QuickXsort} works as follows:
First, choose a pivot element as the median of some sample
(the performance will depend on the size of the sample). Next, partition
the array according to this pivot element, i.\,e., rearrange the array
such that all elements left of the pivot are less or equal and all
elements on the right are greater or equal than the pivot element.  Then, choose one part of the array and sort
it with the algorithm X. After that, sort the other part of the array
recursively with {QuickXsort}.

The main advantage of this procedure is that the part of the array
that is not being sorted currently can be used as temporary memory for
the algorithm X. This yields fast \emph{internal} variants for various
\emph{external} sorting algorithms such as {Mergesort}. The idea is
that whenever a data element should be moved to the extra (additional
or external) element space, instead it is swapped with the data
element occupying the respective position in part of the array which
is used as temporary memory.  Of course, this works only if the
algorithm needs additional storage only for data
elements. Furthermore, the algorithm has to keep track of the
positions of elements which have been swapped. 

For the number of comparisons some general results hold for a wide
class of algorithms X.  Under natural assumptions the average number
of comparisons of X and of \QuickXsort{} differ only by an
$o(n)$-term:
Let X be some sorting algorithm requiring at most $n \log n + cn
+o(n)$ comparisons on average. Then, \QuickXsort{} with a
Median-of-$\sqrt{n}$ pivot selection also needs at most $n \log n + cn
+o(n)$ comparisons on average~\cite{EWCSR}. Sample sizes of approximately $\sqrt{n}$ are likely to be optimal 
\cite{DiekertW13Quick,MartinezR01}.

If the unlikely case happens that always the $\sqrt{n}$
smallest elements are chosen for pivot selection, $\Omega(n^{3/2})$
comparisons are performed. However, as we showed in \cite{EWCSR}, such a worst case is unlikely. Nevertheless, for improving the worst-case complexity, in \cite{EWCSR} we suggested a trick similar to {Introsort} \cite{Mus97} leading to  $n \log n + \Oh(n)$ comparisons in the worst case (use the median of the whole array as pivot if the previous pivot was very bad). In \prettyref{sec:worstcase} of this paper, we refine this method yielding a better average and worst-case performance.

One example for \QuickXsort{} is {QuickMergesort}. For the {Mergesort}
part we use standard (top-down) {Mergesort}, which can be implemented
using $m$ extra element spaces to merge two arrays of length $m$.
After the partitioning, one part of the array -- for a simpler description we assume the first
part -- has to be sorted with {Mergesort} (note, however, that any of the two sides can be sorted with Mergesort as long as the other side contains at least $n/3$ elements. 
In order to do so, the
second half of this first part is sorted recursively with {Mergesort}
while moving the elements to the back of the whole array. The elements
from the back of the array are inserted as dummy elements into the
first part. Then, the first half of the first part is sorted recursively
with {Mergesort} while being moved to the position of the former
second half of the first part. Now, at the front of the array, there is enough space
(filled with dummy elements) such that the two halves can be merged.
The executed stages of the algorithm {QuickMergesort}
(with no median pivot selection strategy applied)  
are illustrated in Fig~\ref{fig:sample}.
\begin{figure}[t]
\begin{center} \hspace{0.12cm}
\begin{tikzpicture}[scale=.9]
{
\draw(-2,4) -- (6.4,4);
\draw(-2,4.7) -- (6.4,4.7);
\draw(-2,4) -- (-2,4.7);
\draw(-1.3,4) -- (-1.3,4.7);
\node[] (2) at (-0.95,4.35) {$11$};
\draw(-0.6,4) -- (-0.6,4.7);
\node[] (3) at (-0.25,4.35) {$4$};
\draw(0.1,4) -- (0.1,4.7);
\node[] (4) at (.45,4.35) {$5$};
\draw(.8,4) -- (.8,4.7);
\node[] (5) at (1.15,4.35) {$6$};
\draw(1.5,4) -- (1.5,4.7);
\node[] (6) at (1.85,4.35) {$10$};
\draw(2.2,4) -- (2.2,4.7);
\node[] (7) at (2.55,4.35) {$9$};
\draw(2.9,4) -- (2.9,4.7);
\node[] (8) at (3.25,4.35) {$2$};
\draw(3.6,4) -- (3.6,4.7);
\node[] (9) at (3.95,4.35) {$3$};
\draw(4.3,4) -- (4.3,4.7);
\node[] (10) at (4.65,4.35) {$1$};
\draw(5.0,4) -- (5.0,4.7);
\node[] (11) at (5.35,4.35) {$0$};
\draw(5.7,4) -- (5.7,4.7);
\node[] (12) at (6.05,4.35) {$8$};
\draw(6.4,4) -- (6.4,4.7);
\node[] (1) at (-1.65,4.35) {${\red 7}$};
}
\end{tikzpicture}

partitioning leads to 

\vspace{0.1cm}
\begin{tikzpicture}[scale=.9]
{
\draw(-2,4) -- (6.4,4);
\draw(-2,4.7) -- (6.4,4.7);
\draw(-2,4) -- (-2,4.7);
\node[] (1) at (-1.65,4.35) {$3$};
\draw(-1.3,4) -- (-1.3,4.7);
\node[] (2) at (-0.95,4.35) {$2$};
\draw(-0.6,4) -- (-0.6,4.7);
\node[] (3) at (-0.25,4.35) {$4$};
\draw(0.1,4) -- (0.1,4.7);
\node[] (4) at (.45,4.35)   {$5$};
\draw(.8,4) -- (.8,4.7);
\node[] (5) at (1.15,4.35)  {$6$};
\draw(1.5,4) -- (1.5,4.7);
\node[] (6) at (1.85,4.35)  {$0$};
\draw(2.2,4) -- (2.2,4.7);
\node[] (7) at (2.55,4.35)  {$1$};
\draw(2.9,4) -- (2.9,4.7);
\draw[thick](3.6,4) -- (3.6,4.7);
\node[] (9) at (3.95,4.35)  {$9$};
\draw(4.3,4) -- (4.3,4.7);
\node[] (10) at (4.65,4.35) {$10$};
\draw(5.0,4) -- (5.0,4.7);
\node[] (11) at (5.35,4.35) {$11$};
\draw(5.7,4) -- (5.7,4.7);
\node[] (12) at (6.05,4.35)  {$8$};
\draw(6.4,4) -- (6.4,4.7);
\node[] (8) at (3.25,4.35) {${\red 7}$};
}



{
\node[]  at (1.5,3.85) {$\underbrace{\hspace{25.2mm}}$};
\node[] (ubxx) at (1.5,3.85) {};
\node[]  at (0.3,3.5) {\small sort recursively};
\node[]  at (3.1,3.5) {\small with Mergesort};

\draw[thick] (.8,4) -- (.8,4.7);
\draw[thick] (2.9,4) -- (2.9,4.7);
}

\begin{scope}[shift={(0,-1.7)}]
{
\draw(-2,4) -- (6.4,4);
\draw(-2,4.7) -- (6.4,4.7);
\draw(-2,4) -- (-2,4.7);
\node[] (1) at (-1.65,4.35) {$3$};
\draw(-1.3,4) -- (-1.3,4.7);
\node[] (2) at (-0.95,4.35) {$2$};
\draw(-0.6,4) -- (-0.6,4.7);
\node[] (3) at (-0.25,4.35) {$4$};
\draw(0.1,4) -- (0.1,4.7);
\node[] (4) at (.45,4.35) {$11$};
\draw(.8,4) -- (.8,4.7);
\node[] (5) at (1.15,4.35) {$9$};
\draw(1.5,4) -- (1.5,4.7);
\node[] (6) at (1.85,4.35) {$10$};
\draw(2.2,4) -- (2.2,4.7);
\node[] (7) at (2.55,4.35) {$8$};
\draw(2.9,4) -- (2.9,4.7);
\node[] (1) at (3.25,4.35) {\red $7$};
\draw[thick](3.6,4) -- (3.6,4.7);
\node[] (9) at (3.95,4.35)  {\blue $0$};
\draw(4.3,4) -- (4.3,4.7);
\node[] (10) at (4.65,4.35) {\blue $1$};
\draw(5.0,4) -- (5.0,4.7);
\node[] (11) at (5.35,4.35) {\blue $5$};
\draw(5.7,4) -- (5.7,4.7);
\node[] (12) at (6.05,4.35) {\blue $6$};
\draw(6.4,4) -- (6.4,4.7);
}
{  
\draw[->] (ubxx) ..controls +(0,-0.9) and (5,5.1).. (5.0,4.7);}
{
  \node[] (ubxxx) at (-0.95,3.85) {};
\node[] (1) at (2,3.48) {\small sort recursively with Mergesort};
\draw[thick] (.8,4) -- (.8,4.7);
\draw[thick] (2.9,4) -- (2.9,4.7);
\node[] (1) at (-0.95,3.85) {$\underbrace{\hspace{18.9mm}}$};}
\end{scope}

\begin{scope}[shift={(0,-3.4)}]

{
\draw(-2,4) -- (6.4,4);
\draw(-2,4.7) -- (6.4,4.7);
\draw(-2,4) -- (-2,4.7);
\node[] (1) at (-1.65,4.35) { $9$};
\draw(-1.3,4) -- (-1.3,4.7);
\node[] (2) at (-0.95,4.35) { $10$};
\draw(-0.6,4) -- (-0.6,4.7);
\node[] (3) at (-0.25,4.35) { $8$};
\draw(0.1,4) -- (0.1,4.7);
\node[] (4) at (.45,4.35) {$11$};
\draw(.8,4) -- (.8,4.7);
\node[] (5) at (1.15,4.35) {\darkgreen$2$};
\draw(1.5,4) -- (1.5,4.7);
\node[] (6) at (1.85,4.35) {\darkgreen$3$};
\draw(2.2,4) -- (2.2,4.7);
\node[] (7) at (2.55,4.35) {\darkgreen $4$};
\draw(2.9,4) -- (2.9,4.7);
\node[] (1) at (3.25,4.35) {\red $7$};
\draw[thick](3.6,4) -- (3.6,4.7);
\node[] (9) at (3.95,4.35)  {\blue $0$};
\draw(4.3,4) -- (4.3,4.7);
\node[] (10) at (4.65,4.35) {\blue $1$};
\draw(5.0,4) -- (5.0,4.7);
\node[] (11) at (5.35,4.35) {\blue $5$};
\draw(5.7,4) -- (5.7,4.7);
\node[] (12) at (6.05,4.35) {\blue $6$};
\draw(6.4,4) -- (6.4,4.7);
}

{  
\draw[->] (ubxxx) ..controls +(0,-0.9) and (1.85,5.1).. (1.85,4.7);}

%

{\node[] (1) at (1.85,3.85) {$\underbrace{\hspace{18.9mm}}$};
\node[] (1) at (5.0,3.85) {$\underbrace{\hspace{25.2mm}}$};
  \node[] (ubxxxx) at (1.85,3.85) {};
    \node[] (ubxxxxx) at (5,3.85) {};
\node[] (uba) at (-1,3.38) {\small merge two parts};
}

\end{scope}

\begin{scope}[shift={(0,-5.1)}]

{
\draw[-] (ubxxxxx) ..controls +(0,-0.5) and (.45,5.4) .. (.45,4.9);
\draw[-] (ubxxxx) ..controls +(0,-0.5) and (.45,5.4) .. (.45,4.9);
\draw[->] (.45,5)-- (.45,4.7);
}

{
\draw(-2,4) -- (6.4,4);
\draw(-2,4.7) -- (6.4,4.7);
\draw(-2,4) -- (-2,4.7);

\draw(-1.3,4) -- (-1.3,4.7);

\draw(-0.6,4) -- (-0.6,4.7);

\draw(0.1,4) -- (0.1,4.7);

\draw(.8,4) -- (.8,4.7);
\draw(.8,4) -- (.8,4.7);

\draw(1.5,4) -- (1.5,4.7);

\draw(2.2,4) -- (2.2,4.7);

\draw(2.9,4) -- (2.9,4.7);

\draw[thick](3.6,4) -- (3.6,4.7);

\draw(4.3,4) -- (4.3,4.7);

\draw(5.0,4) -- (5.0,4.7);

\draw(5.7,4) -- (5.7,4.7);

\draw(6.4,4) -- (6.4,4.7);
}
{

  \node[] (6) at (-1.65,4.35) {\blue $0$};
    \node[] (7) at (-0.95,4.35) {\blue  $1$};
      \node[] (1) at (-0.25,4.35) {\darkgreen $2$};
      \node[] (4) at (.45,4.35) {\darkgreen $3$};
      \node[] (5) at (1.15,4.35) {\darkgreen $4$};
      \node[] (1) at (1.85,4.35) {\blue $5$};
        \node[] (9) at (2.55,4.35) {\blue $6$};
\node[] (10) at (3.25,4.35) {\red $7$};
  \node[] (11) at (3.95,4.35) {$11$};
   \node[] (12) at (4.65,4.35) {$9$};
  \node[] (2) at (5.35,4.35) {$8$};
  \node[] (3) at (6.05,4.35) {$10$};

}

{\node[] (1) at (5.0,3.85) {$\underbrace{\hspace{25.2mm}}$};
\node[] (1) at (3.25,3.5) {\small sort recursively with \QuickMergesort};}

\end{scope}
\end{tikzpicture}%
\end{center}
\caption{Example for the execution of {QuickMergesort}.}
\label{fig:sample}
\end{figure}
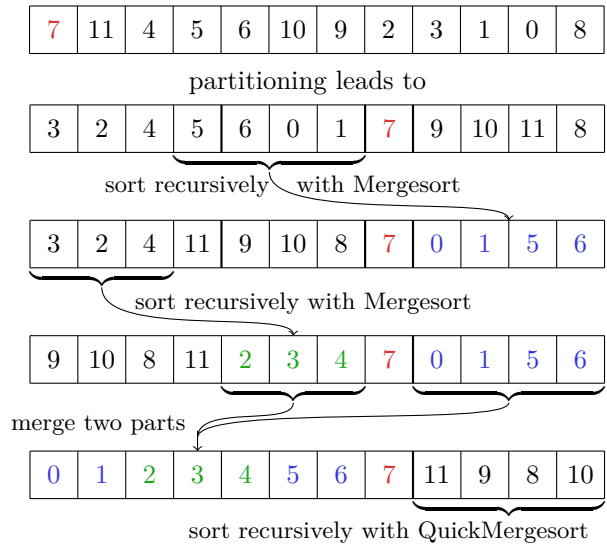

{Mergesort} requires approximately $n \log n - 1.26n$ comparisons
on average, so that with a Median-of-$\sqrt{n}$ we
obtain an internal sorting algorithm with $n \log n - 1.26n + o(n)$
comparisons on average.  One can do even better by sorting small subarrays with a more complicated algorithm requiring less comparisons~-- for details see \cite{EWCSR}.

Since the Median-of-3 variant (i.\,e.\ CleverQuickMergesortsort) shows a slightly better practical performance than with Median-of-$\sqrt{n}$ (see \cite{EWCSR}), we provide here a theoretical analysis of it by showing  that 
{CleverQuickMergesortsort} performs at most $n \log n - 0.75n + o(n)$ 
comparisons on average. 
In fact, as in \cite{EWCSR} we show a more general result for CleverQuickXsort for an arbitrary algorithm X.

\begin{theorem} [Average Case {CleverQuickXsort}] \label{thm:CQXS}
  Let the algorithm {X} perform at most $\alpha n \log n + cn + \Oh(\log n)$
  comparisons on average. Then, {CleverQuickXsort} performs at most $\alpha n \log n + (c + \kappa_\alpha) n + \Oh(\log^2
  n)$ comparisons on average with $\kappa_\alpha = \frac{4}{15}\left(12 - \frac{7\alpha}{\ln 2}\right) \leq 0{.}51$.
\end{theorem}

Since {Mergesort} requires at most $n \log n - 1.26n + o(n)$ comparisons on average, we obtain the following corollary: 
\begin{corollary} [Average Case {CleverQuickMergesort}] \label{cor:CQMS}
{CleverQuickMergesort} is an in-place algorithm that performs at most $n \log n - 0.75n + o(n)$ 
comparisons on average.
\end{corollary}

\begin{proof}[Proof of \prettyref{thm:CQXS}]
The probability of choosing the $k$-th element (in the ordered sequence) as pivot of a random $n$-element array is $\Pr{\!\text{pivot }= k\!} = (k-1)(n-k){\binom{n}{3}}^{-1}$ (one element of the three element set has to be less than the $k$-th, one equal to the $k$-th, and one greater than $k$-th element of the array). Note that this holds no matter whether we select the three elements at random or we use fixed positions and average over all input permutations. Since probabilities sum up to 1, we have
	\begin{align}
	\sum_{k=1}^{n}(k-1)(n-k){\binom{n}{3}}^{-1} = 1. \label{eq:prob}
	\end{align}

		Moreover, partitioning preserves randomness of the two sides of the array~-- this includes the positions where the other two elements from the pivot sample are placed (since for a fixed pivot, every element smaller (resp.\ greater) than the pivot has the same probability of being part of the sample). Also, using the array as temporary space for Mergesort does not destroy randomness since the dummy elements are never compared. 
	
	Let $T(n)$ be the average-case number of comparisons of 
{CleverQuickXsort} for sorting an input of size $n$ 
 and let $$S(n)= \alpha n \log n + cn + d (1+\log n)$$ be a bound for the average number of comparisons
  of the algorithm X (e.\,g.\ Mergesort).
We will show by induction that $$T(n) \leq \alpha n\log n + (c+ \kappa_\alpha )n  +  D (1+\log^2 n) $$ for some constant $D\geq d$ (which we specify later such that the induction base is satisfied) and
$\kappa_\alpha = \frac{4}{15}\left(12 - \frac{7\alpha}{\ln 2}\right) \leq 0.51$ (since $\alpha\geq 1$ by the general lower bound on sorting).
As induction hypothesis
for $1\leq k \leq n$ we assume that 
\begin{align*}
\max\{\, T(k-1) & + S(n-k), T(n-k)+S(k-1)\,\}\\
&\leq \alpha (k-1) \log (k-1) + \alpha(n-k) \log(n-k) + cn + \kappa_\alpha \max\{k-1, n-k\}\\ 
&\qquad + D\left(1+\log^2 (\max\{k-1, n-k\})\right) + d\left(1+\log (\min\{k-1, n-k\})\right) \\  &=: f(n,k).
\end{align*}

 In order to find the pivot element, three comparisons are needed. After that, for partitioning $n-3$ comparisons are performed (all except the three elements of the pivot sample are compared with the pivot). Since after partitioning, one part of the array is sorted with X and the other recursively with CleverQuickXsort, we obtain the recurrence
 {\allowdisplaybreaks
\begin{align}
\nonumber	T(n) &\leq n + \sum_{k=1}^n \Pr{\!\text{pivot }= k\!}\cdot\max\left\{\, T(k-1) + S(n-k), T(n-k)+S(k-1)\,\right\}\\
\nonumber	&\leq n + \sum_{k=1}^n \frac{(k-1)(n-k)}{\binom{n}{3}}
	f(n,k)
	\\
	&\leq n
        + \frac{1}{\binom{n}{3}}\sum_{k=1}^n (k-1)(n-k) \Bigl(\alpha (k-1)  \log (k-1) +	  \alpha(n-k) \log(n-k)\Bigr) \label{eq:nlogn}\\ 
	&\qquad +  \frac{1}{\binom{n}{3}}\sum_{k=1}^n (k-1)(n-k) \kappa_\alpha \max\{k-1, n-k\}  \label{eq:ckappa}\\  
	&\qquad + \frac{1}{\binom{n}{3}}\sum_{k=1}^n (k-1)(n-k) D\log^2(\max\{k-1, n-k\}) \label{eq:logsquare}\\ 
	&\qquad +  \frac{1}{\binom{n}{3}}\sum_{k=1}^n (k-1)(n-k) \Bigl(c n +   D + d +d\log\left( \min\{k-1, n-k\}\right)\Bigr)\label{eq:rest}
\end{align}}%
We simplify the terms \prettyref{eq:nlogn}--\prettyref{eq:rest} separately using \texttt{http://www.wolframalpha.com/} for evaluating the sums and integrals.
The function $x \mapsto g(x) = (x-1)^2(n-x) \log (x-1) $ is non-negative and has a single maximum for $1 \leq x \leq n$ at position $x=\xi$; on the left of $\xi$, it is monotonically increasing, on the right monotonically decreasing. Therefore, 
\begin{align*}
\sum_{k=1}^n g(k) &= \sum_{k=1}^{\floor{\xi}} g(k) + \sum_{k=\floor{\xi} + 1}^n g(k)\
\leq \int_{1}^{\floor{\xi}} g(x)\,d x + \int_{\floor{\xi} + 1}^n g(x)\,d x + 2g(\xi).
\end{align*}
Since the second term of \prettyref{eq:nlogn} is obtained from the first one by a substitution $k \mapsto n + 1 - k$, it follows that 
\begin{align*}
\prettyref{eq:nlogn} -n & \leq \frac{\alpha}{\binom{n}{3}}\cdot\Biggl(\int_{1}^{n} g(x)\,d x + \int_{1}^n g(n + 1 - x)\,d x + 4g(\xi)\Biggr)\\
&\leq  \frac{\alpha}{\binom{n}{3}}\cdot\Biggl(2\int_{1}^n x^2(n-x) \log x \,d x + 4 n^3\log n\Biggr)\\
&=
 \frac{\alpha}{\binom{n}{3}}\cdot\Biggl(\frac{2}{144 \ln 2}n^4(12\ln n-7) + 4 n^3\log n\Biggr) \leq \alpha n\log n-\frac{7\alpha}{12 \ln 2}n + c_3\alpha\log n
\end{align*}
for some properly chosen constant $c_3$.
Now, first assume that $\kappa_\alpha \geq 0$. Then we have
\begin{align*}
\prettyref{eq:ckappa} & \leq\frac{2\kappa_\alpha}{\binom{n}{3}}\sum_{k=1}^{\ceil{n/2}} (k-1)(n-k)^2 \leq \frac{2\kappa_\alpha n}{192\binom{n}{3}}\left(11n^3 - 20n^2 - 44n + 80 \right) \leq \frac{11}{16}\kappa_\alpha n + c_4
\end{align*}
for some constant $c_4$.
On the other hand, if $\kappa_\alpha < 0$, we have
\begin{align*}
\prettyref{eq:ckappa} & \leq\frac{2\kappa_\alpha}{\binom{n}{3}}\sum_{k=1}^{\floor{n/2}} (k-1)(n-k)^2 \leq \frac{2\kappa_\alpha n}{192\binom{n}{3}}\left(11n^3 - 68n^2 - 100n + 16 \right) \leq \frac{11}{16}\kappa_\alpha n + c_4
\end{align*}
for some constant $c_4$. Thus, in any case, we have $\prettyref{eq:ckappa} \leq \frac{11}{16}\kappa_\alpha n + c_4$. With the same argument as for \prettyref{eq:nlogn}, we have
\begin{align*}
\prettyref{eq:logsquare} &\leq \frac{2D}{\binom{n}{3}}\sum_{k=\floor{n/2}}^{n} (k-1)(n-k)\log^2(k-1)\\
&\leq \frac{2D}{\binom{n}{3}}\int_{\floor{n/2}}^{n} (x-1)(n-x)\log^2(x-1)\,dx + D\cdot c_5' \leq D \log^2 n - \frac{5D}{3}\log n + D \cdot c_5
\end{align*}
for some constants $c_5'$ and $c_5$.
Finally, by \prettyref{eq:prob}, we have
\begin{align*}
\prettyref{eq:rest} \leq c n +   D + d + d \log(n/2) = cn + D + d\log n.
\end{align*}
Now, we combine all the terms and obtain
\begin{align*}
T(n) 
&\leq \alpha n \log n + n\left(1+ \frac{-7\alpha}{12\ln 2} + c + \frac{11}{16}\kappa_\alpha  
\right)\\
&\qquad + c_3\alpha\log n + c_4 + D\log^2n  - \frac{5D}{3}\log n +  Dc_5 + D +  d\log n
\end{align*}
We can choose $D$ such that $ \frac{5D}{3}\log n \geq  c_3\alpha\log n + c_4 +  Dc_5 + D +  d\log n$ for $n$ large enough and $D \geq T(n)$ for all smaller $n$. Hence, we conclude the proof of \prettyref{thm:CQXS}:
\begin{align*}
T(n) 
&\leq \alpha n \log n + n\left(1+ \frac{-7\alpha}{12\ln 2} + c + \frac{11}{16}\kappa_\alpha  \right) + D\log^2n + D\\
& =  \alpha n \log n + n\left(1+ \frac{-7\alpha}{12\ln 2} + c + \frac{11}{16}\cdot\frac{4}{15}\left(12 - \frac{7\alpha}{\ln 2}\right)  \right) + D\log^2n + D\\
&= \alpha n \log n + (c + \kappa_\alpha )n + D\log^2n + D.
\end{align*}
\end{proof}

Notice that in the case that in each recursion level always the smaller part is sorted with X, the inequalities in the proof of \prettyref{thm:CQXS} are tight up to some lower order terms. Thus, the proof can be easily modified to provide a lower bound of $\alpha n \log n + (c + \kappa_\alpha) n - \Oh(\log^2
n)$ comparisons in this special case.

\section{{QuickMergeXsort}}\label{sec:QMQS}

{QuickMergeXsort} agrees with {QuickMergesort} up to the following change: for arrays of size smaller than
some threshold cardinality {X\_THRESH}, the sorting algorithm X is called  (instead of Mergesort)
and the sorted elements are moved to 
their target location expected by
{QuickMergesort}. 

Fig.~\ref{fig:pseudocode} provides the full implementation details of
{QuickMerge(X)sort} (in {C++}). The realization of the sorting
algorithm X and the partitioning algorithm have to be added. The listing shows that by dropping the base cases
from {QuickMergesort} the code is short enough for textbooks on
algorithms and data structures. The general principle is that we have
a merging step that takes two sorted areas, merges and swaps them
into a third one.

The program \emph{msort} applies {Mergesort} with X
as a stopper. It goes down the recursion tree and shrinks the size of
the array accordingly.  If the array is small enough, the algorithm
calls X followed by a joint movement (memory copy) of array elements
(the only change of code wrt.\ {QuickMergesort}).
The algorithm \emph{out} serves as an interface between 
the recursive procedure \emph{msort} and top-level procedure \emph{sort}. 
Last, but not least, we have the overall internal sorting 
algorithm \emph{sort}, that performs the partitioning.
\lstset{ %
  backgroundcolor=\color{white},   
  basicstyle=\footnotesize,        
  breakatwhitespace=false,         
  breaklines=true,                 
  captionpos=b,                    
  commentstyle=\color{mygreen},    
  deletekeywords={...},            
  escapeinside={\%*}{*)},          
  extendedchars=true,              
  frame=single,	                   
  keepspaces=true,                 
  keywordstyle=\color{blue},       
  language=Octave,                 
  otherkeywords={*,...},           
  numbers=left,                    
  numbersep=5pt,                   
  numberstyle=\tiny\color{mygray}, 
  rulecolor=\color{black},         
  showspaces=false,                
  showstringspaces=false,          
  showtabs=false,                  
  stepnumber=1,                    
  stringstyle=\color{mymauve},     
  tabsize=10,	                   
  title=\lstname                   
}
\begin{figure}[t]
  \begin{center}
    \begin{minipage}{14cm}
  \begin{lstlisting}[language = C]
 typedef std::vector<t>::iterator iter; 
 void merge(iter begin1, iter end1, iter target, iter endtarget) {
    iter i1 = begin1, i2 = target + (end1 - begin1), ires = target;
    t temp = *target;
    while (i1 != end1 && i2 != endtarget) {
      iter tempit = (*i1 < *i2) ? i1++ : i2++;
      *ires++ = *tempit; *tempit = *ires;
    }
    while(i1 < end1) { *ires++ = *i1; *i1++ = *ires; }
    *(i1 - 1) = temp;
  }
  void msort(iter begin, iter end, iter target)  {
    index n = end - begin;
    if (n < X_THRESH) {
      X(begin, end);
      for(int i=0; i<n; i++) std::swap(begin[i], target[i])
    }
    else {
      index q = n / 2;
      msort(begin + q, end, target + q);
      msort(begin, begin + q, begin + q);
      merge(begin + q, begin + n , target, target + n);
    }
  }
  void out(iter begin, iter end, iter temp) {  
    index n = end - begin;
    if (n > 1) {  
      index q = n / 2, r = n - q;
      msort(begin + q, end, temp);
      msort(begin, begin + q, begin + r);
      merge(temp , temp + r , begin, end);
    }
  }
  void sort(std::vector<t> &a) {
    iter begin = a.begin(), end = a.end();      
    while (begin < end) {
      iter b = partition(begin,end);
      if (b < (end + begin)/2) { out(begin, b, b+1); begin = b+1; } 
      else { out(b+1, end, begin); end = b; }
    }
  }
  \end{lstlisting}
  \end{minipage}
\end{center}
  \vspace{-1cm}
\caption{Implementation of QuickMergeXsort.}
\label{fig:pseudocode}
\end{figure}
The following result is a generalization of the $1
\cdot n \log n + cn +o(n)$ average comparisons bound in~\cite{EWCSR}.
Indeed, the proof is almost a verbatim copy of the proof of \cite[Thm.\ 1]{EWCSR} 
(compare to the role of $\alpha$ in the proof of \prettyref{thm:CQXS}). 
\begin{theorem} [Average-Case \QuickXsort{}]
\label{thm:quickXsort}
For $\alpha \ge 1$ let X be some sorting algorithm requiring at most
$\alpha \cdot n \log n + cn +o(n)$ comparisons on average. Then,
\QuickXsort{} with a Median-of-$\sqrt{n}$ pivot selection 
also needs at most $\alpha \cdot n \log n + cn
+o(n)$ comparisons on average. 
\end{theorem}

We are now ready to analyze the average-case performance of
{QuickMergeXsort}.

\begin{theorem} [Average-Case {QuickMergeXsort}/{CleverQuickMergeXsort}]
\label{thm:quickmergeXsqrtsortac}
Let X be a sorting algorithm with $\alpha \cdot n \log n + c n +
o(n)$ comparisons in the average case, called when reaching $\ceil{n^\beta}$
elements, $0 <\beta < 1$. 
Then, \QuickMergeXsort{} with Median-of-$\sqrt{n}$ pivot
selection, as well as with Median-of-3
pivot selection, is a sorting algorithm that needs at most
$(\alpha \beta + (1 - \beta)) \cdot n \log n + \Oh(n)$
comparisons in the average case.
\end{theorem}

\begin{proof} To begin with we analyze {MergeXsort}, i.e., Mergesort, with
  recursion stopper X.  
%
We assume that every path of the recursion tree of Mergesort has the same length until the algorithm switches to X. This can be easily implemented and guarantees that all calls to X are made on arrays of almost identical size.

First, we look at the $\ceil{(\log n) \cdot (1-\beta)}$ top layers of the
recursion tree, which are sorted by {Mergesort}.  In the worst-case, in layer $i$
of the tree, {Mergesort} requires at most $n-2^i < n$ comparisons, so
that in total we have at most
$C_{\text{MergeXsort}} (n) 
= n \cdot \ceil{(1-\beta) \cdot \log n}$ 
element comparisons. 
The average case differs only negligibly from the worst case.

In the $\ceil{(\log n) \cdot (1-\beta)}$ recursion levels of Mergesort, $2^{\ceil{(1-\beta)\log n)}}$ sorted arrays are merged to one large sorted array. Each of the $g_{\beta}(n) = 2^{\ceil{(1-\beta)\log n}}$ arrays is of size at most $f_\beta(n)= \ceil{2^{\log n - \ceil{(1-\beta)\log n}}} \leq \ceil{n^\beta}$.

Next,  we look at the $g_{\beta}(n) = 2^{\ceil{(1-\beta)\log n}}$ calls to X. Let $C_X(n)$ denote the average number of element comparisons executed by all calls of X. 
Given that $g_{\beta}(n) f_{\beta}(n) = n+ \Oh(n^{1-\beta})$ and 
$\log f_\beta(n) = \log \ceil{2^{\log n - \ceil{(1-\beta)\log n}}} = \log n - \ceil{(1-\beta)\log n} + \Oh(1/n^\beta)$,
we obtain
%
\begin{align*}
C_{X} (n) &= g_{\beta}(n) \cdot (\alpha \cdot f_\beta(n)  \log f_\beta(n) + c f_\beta(n)  + o(f_\beta(n) ))  \\
&=  \alpha \cdot n  \log f_\beta(n) + cn  + o(n)  
= \alpha \cdot n \left(\log n - \ceil{(1-\beta)\log n}\right) + c n + o(n)
\end{align*}
In cumulation, for the
average-case number of comparisons of {MergeXsort} we have the following upper bound \begin{align*}
C_{\text{MergeXsort}}(n) &= C_{X}
(n) + C_{\text{MergeXsort}} (n)\\
& \leq \alpha \cdot n \left(\log n - \ceil{(1-\beta)\log n}\right) + c n +
o(n) + n \ceil{(1-\beta) \log n} \\
&=n \bigl(\alpha \cdot \log n - (\alpha - 1) \ceil{(1-\beta)\log n}\bigr) + c n +
o(n) \\
&= (\alpha\beta+(1-\beta)) \cdot n
\log n + \Oh(n).
\end{align*}


Using Theorem~\ref{thm:quickXsort} (resp.\ \prettyref{thm:CQXS} for Median-of-3) we obtain
the matching bound of at most $(\alpha\beta+(1-\beta)) \cdot n \log n + 
\Oh(n)$ element comparisons on average for {QuickMergeXsort}.
\end{proof}

Theorem~\ref{thm:quickmergeXsqrtsortac} 
implies that CleverQuickMergeXsort{} implemented with 
{CleverQuicksort} as recursion stopper
at $\sqrt{n}$ elements ($\beta = 1/2$
) is a sorting
algorithm that needs at most 
$((\alpha+1)/2) \cdot n \log n + \Oh(n) = 1.094 \cdot n \log n +\Oh(n)$
comparisons on average.

\section{Worst-Case Efficient QuickMergeSort}\label{sec:worstcase}

Although QuickMergesort has an $\Oh(n^2)$ worst-case running time, is is quite simple to guarantee a worst-case number of comparisons of $n \log n + \Oh(n)$: just choose the median of the whole array as pivot. This is essentially how InSituMergesort~\cite{ElmasryKS12} works. The most efficient way for finding the median is using Quickselect \cite{Hoare61_find} as applied in InSituMergesort. However, this does not allow the desired bound on the number of comparisons (even not when using IntroSelect as in~\cite{ElmasryKS12}). Alternatively, one could use the median-of-medians algorithm \cite{BFPRT73}, which, while having a linear worst-case running time, on average is quite slow. In this section we describe a slight variation of the median-of-medians approach, which combines a linear worst-case running time with almost the same average performance as InSituMergesort.

Again, the crucial observation is that it is not necessary to use the actual median as pivot. As remarked in \prettyref{sec:quickXsort}, the larger of the two sides of the partitioned array can be sorted with Mergesort as long as the smaller side contains at least one third of the total number of elements. Therefore, it suffices to find a pivot which guarantees such a partition. For doing so, we can apply the idea of the median-of-medians algorithm: for sorting an array of $n$ elements, we choose first $n/3$ elements as median of three elements each. Then, the median-of-medians algorithm is used to find the median of those $n/3$ elements. This median becomes the next pivot. Like for the median-of-medians algorithm \cite{BFPRT73}, this ensures that at least $2\cdot\floor{n/6}$ elements are less or equal and at least the same number of elements are greater or equal than the pivot~-- thus, always the larger part of the partitioned array can be sorted with Mergesort and the recursion takes place on the smaller part.
The big advantage over the straightforward application of the median-of-medians algorithm it that it is called on an array of only size $n/3$ (with the cost of introducing a small overhead for finding the $n/3$ medians of three)~-- giving less weight on its big constant for the linear number of comparisons. 
We call this algorithm MoMQuickMergesort (MOMQMS).

  In our implementation of the median-of-medians algorithm, we use select the pivot as median of the medians of groups of five elements~-- we refer to \cite[Sec.\ 9.3]{CLRS09} for a detailed description.
  The total number $T(n)$ of comparisons in the worst case of MoMQuickMergesort is bounded by
$$T(n) \leq T(n/2) + S(n/2) + M(n/3) + \frac{n}{3}\cdot 3 + \frac{2}{3}n$$
where $S(n)$ is the number of comparisons incurred by Mergesort and $M(n)$ the number of comparisons for the median-of-medians algorithm. We have $M(n) \leq 22n$ (for the variant used in our implementation, which uses seven comparisons for finding the median of five elements). The $\frac{n}{3}\cdot 3$-term comes from finding $n/3$ medians of three elements, the $2n/3$ comparisons from partitioning the remaining elements (after finding the pivot, the correct side of the partition is known for $n/3$ elements).

Since by  \cite{Knu73} we have $S(n) \leq n\log n - 0.9n$, this yields
$$T(n) \leq T(n/2) + \frac{n}{2}\log(n/2) - \frac{0.9 n}{2}+ \frac{22}{3}n + \frac{5}{3}n$$
resolving to $T(n) \leq n\log n + 16.1 n$. 

For our implementation we also use a slight improvement over the basic median-of-medians algorithm by using the approach of adaption, which was first introduced in \cite{MartinezPV10} for Quickselect and recently applied to the median-of-medians algorithm \cite{Alexandrescu17}. More specifically, whenever in a recursive call the $k$-th element is searched with $k$ far apart from $n/2$ (more precisely for $k \leq 0.3n$ or $k\geq 0.7n$), we do not choose the median of the medians as pivot but an element proportional to $k$ (while still guaranteeing that at least $0.3n$ elements are discarded for the next recursive call as in \cite{BFPRT73}).

 Notice that in the presence of duplicate elements, we need to apply three-way partitioning for guaranteeing that worst-case number of comparisons (that is elements equal to the pivot are placed in the middle and not included into the recursive call nor into Mergesort). With the usual partitioning (as in our experiments), we obtain a worse bound for the worst case since it might happen that the smaller part of the array has to be sorted with Mergesort. 

In order to achieve the guarantee for the worst case together with the efficiency of the Median-of-3 pivot sampling, we can combine the two approaches using a  trick similar to {Introsort} \cite{Mus97}:  we fix some small $\delta >0$.  Whenever the pivot is
contained in the interval $\left[\delta n, (1-\delta)n \right]$, the next pivot is selected as  Median-of-3, otherwise according to the worst-case efficient procedure described in the previous section~-- for the following pivots switch back to Median-of-3. When choosing  $\delta$ not too small, the worst case number of comparisons will be only approximately $2n$ more than of MoMQuickMergesort (because in the worst case before every partitioning step according to MoMQuickMergesort, there will be one partitioning step with Median-of-3 using $n$ comparisons), while the average is almost as CleverQuickMergesort. We propose $\delta = 1/16$.
We call this algorithm HybridQuickMergsort (HQMS). 

\section{Experiments}\label{sec:experiments}

The collection of sorting algorithms we considered for comparison is
much larger than the one we present here, but the bar of being
competitive wrt.\ state-of-the-art library implementations in
{C++} and {Java} on basic data types is surprisingly
high. For example, all {Heapsort} variants we are aware of 
fail this test, we checked refined implementations of
{Binary Heapsort}~\cite{Flo64,Wil64},
{Bottom-Up Heapsort}~\cite{Weg93}, {MDR Heapsort} \cite{Wegener:MDR},
{QuickHeapsort}~\cite{DiekertW13Quick}, and
{Weak-Heapsort}~\cite{Dut93}.
Some of these algorithm even use
extra space.
{Timsort} (by Tim Peters; used in {Java} for
sorting non-elementary object sequences) was less performant on simple data types.

There are fast algorithms that exploit the set of keys to be sorted
(like {CountingSort} or {Radixsort}), but we aim at a general 
algorithm.

We also experimented with Sanders and Winkel's
{SuperScalarSampleSort} that has a particular memory
profile~\cite{SandersW04,EdelkampW16_BQS}.  
The main reason not to include the
results was that it allocates substantial amounts of space for
the elements and, thus, is not internal.
We experienced that it acts fast on random data,
but not as good on presorted inputs.

One remaining competitor was {(Bottom-Up) Mergesort}
(\texttt{std::stable\_sort}) in the {C++} STL library, which on some
inputs shows a very good performance. As this is an external
algorithm, we chose a tuned version of in-place {Mergesort}
(\texttt{stl::inplace\_stable\_sort} simply was too slow) called
{InSituMergesort} (ISMS)~\cite{ElmasryKS12} for our experiments.

According
to~\cite{dqsanalysis1,dqsanalysis2,dqsanalysis3}, for the
{DualPivotQuicksort} algorithm variants, there was no
clear-cut winner, but the experiments suggested that the
standard ones had a slight edge.  For {DualPivotQuicksort}
we translated the most recent Oracle's (Java) version (the
  algorithm selects the 2nd and 4th element of the inner five pivot
  candidates of a split-into-7).
  As the full sorting
  algorithm is lengthy and contains many checks for special input
  types (with different code fragments and parameter settings for
  sorting arrays of bytes, shorts, ints, floats, doubles etc.) we
  extracted the integer part.

{TunedQuicksort}~\cite{ElmasryKS12} is an engineered implementation of
{CleverQuicksort}, probably unnoticed by the public and contained in a
paper on tuning {Mergesort} for studying branch misprediction as
in~\cite{Kaligosi}. It applies Lomuto's uni-directional Median-of-3
partitioner~\cite{CLRS09}, which works well for permutations and a
limited number of duplicates in the element set. As with {Introsort},
the algorithm stops recursion, if less than a fixed number of elements
are reached (16 in our case). These elements are then sorted together,
calling STL's {Insertionsort} algorithm. The implementation utilizes a
stack to avoid recursion, being responsible for tracking the remaining
array intervals to be processed. 
We dropped TunedQuicksort from the experiments
as it failed on presorted data and data with duplicates,
but we used parts of its efficient stack-based
implementation. This advanced CleverQuicksort implementation
and {CleverQuickMergesort} (QMS) are the two
extremes, while {CleverQuickMergeCleverQuicksort} 
(QuickMergeCleverQuicksort with a modified TunedQuicksort
implementation at $\sqrt{n}$
elements) (QMQS for short) is our tested intermediate.

QMS uses
hard-coded base cases for $n<10$, while the recursion stopper
in QMQS does
not.  Depending on the size of the arrays the
displayed numbers are averages over multiple runs
(repeats)\footnote{Experiments were run on one core of an Intel Core
  i5-2520M CPU (2.50GHz, 3MB Cache) with 16GB RAM; Operating system:
  Ubuntu Linux 64bit; Compiler: GNU's \texttt{g++} (4.8.2); optimized
  with flags \texttt{-O3 -march=native -funroll-loops}.  }.  The
arrays we sorted were random permutations of $\{1,\ldots,n\}$.  The
number of element comparisons was measured by increasing a counter for
every comparison.

For CPU time experiments we used vectors of integers as this is often most challenging for algorithms
with a lower number of comparison. All algorithms
sort the same arrays.  As counting the number of comparisons affects
the speed of the sorting algorithms, for further measurements (e.g., moves
and comparisons) we started another sets of experiments. 

We made element comparisons more expensive 
(we experimented with logarithms, and elements as vectors and records). 
Through a lower number of comparisons results were even better.

As a first empirical observation, for {Introsort} (Std) the number of
element comparisons divided by $n \log n$ is larger than $1.18$, due
to higher lower-order terms.  As theoretically shown, for QMS the
number of element comparisons divided by $n \log n$ was below~1.

For our QuickMergesort implementations we used the block partitionioner from \cite{EdelkampW16_BQS}, which improves the performance considerably over the standard Hoare partitioner.
Figs.~\ref{fig:sortrandom}--\ref{fig:sortrandomdup} show the results 
when sorting
random integer data (with QMQS: CleverQuickMergeCleverQuicksort, QMS:
CleverQuickMergesort,
MOMQMS: worst-case-efficent QuickMergesort,
HQMS: hybrid of worst-case- and average-case-efficient QuickMergeSort,
ISMS: InSituMergesort,
Java: DualPivotQuicksort,
and Std: \texttt{std::sort}). 
Times displayed are the
total running times divided by the number of elements (in ns).  We 
see that QuickMergeSort variants are fast.
For measuring element moves (assignments of input
data elements, e.g., a swap of two elements is counted as three
moves).

\begin{figure}[tb]
\begin{center}
\includegraphics[width=6.8cm]{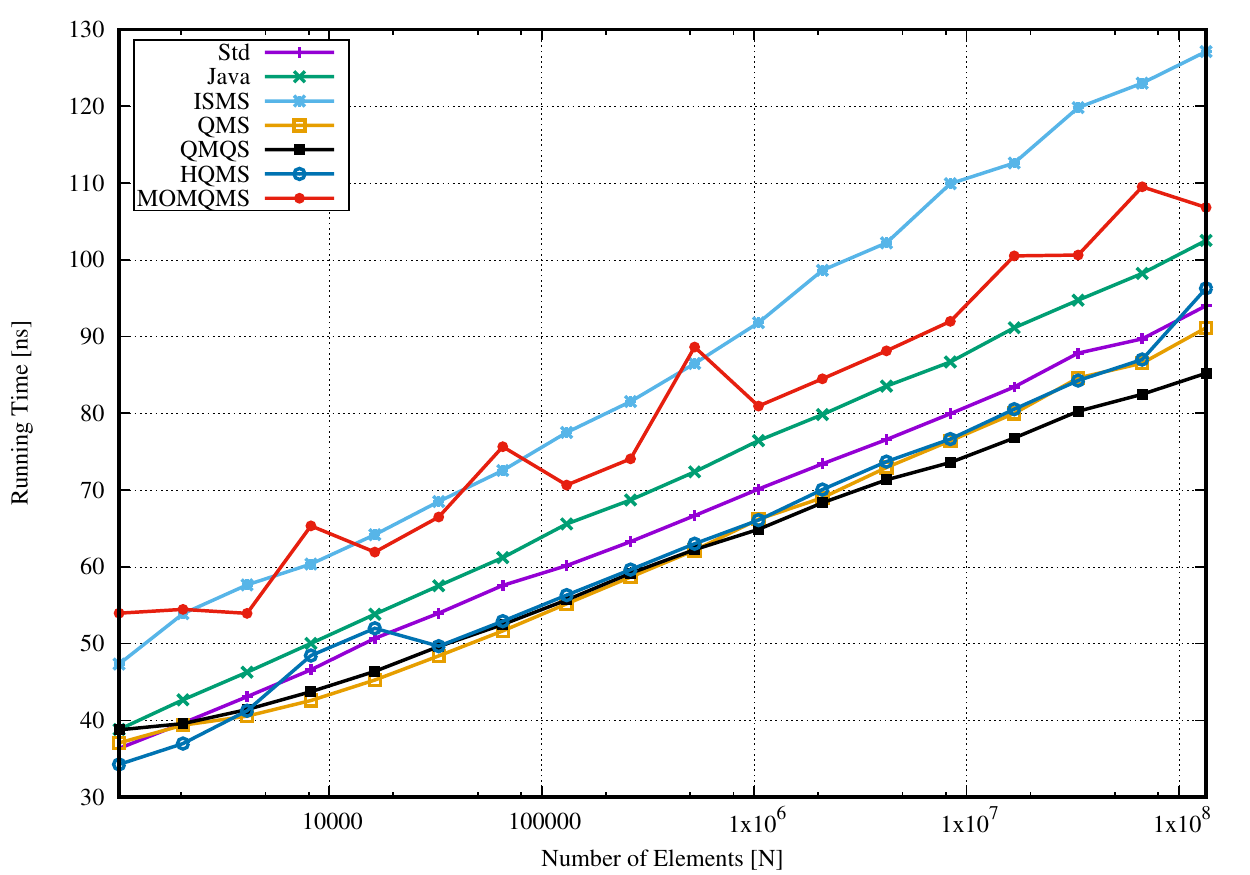}
\includegraphics[width=6.8cm]{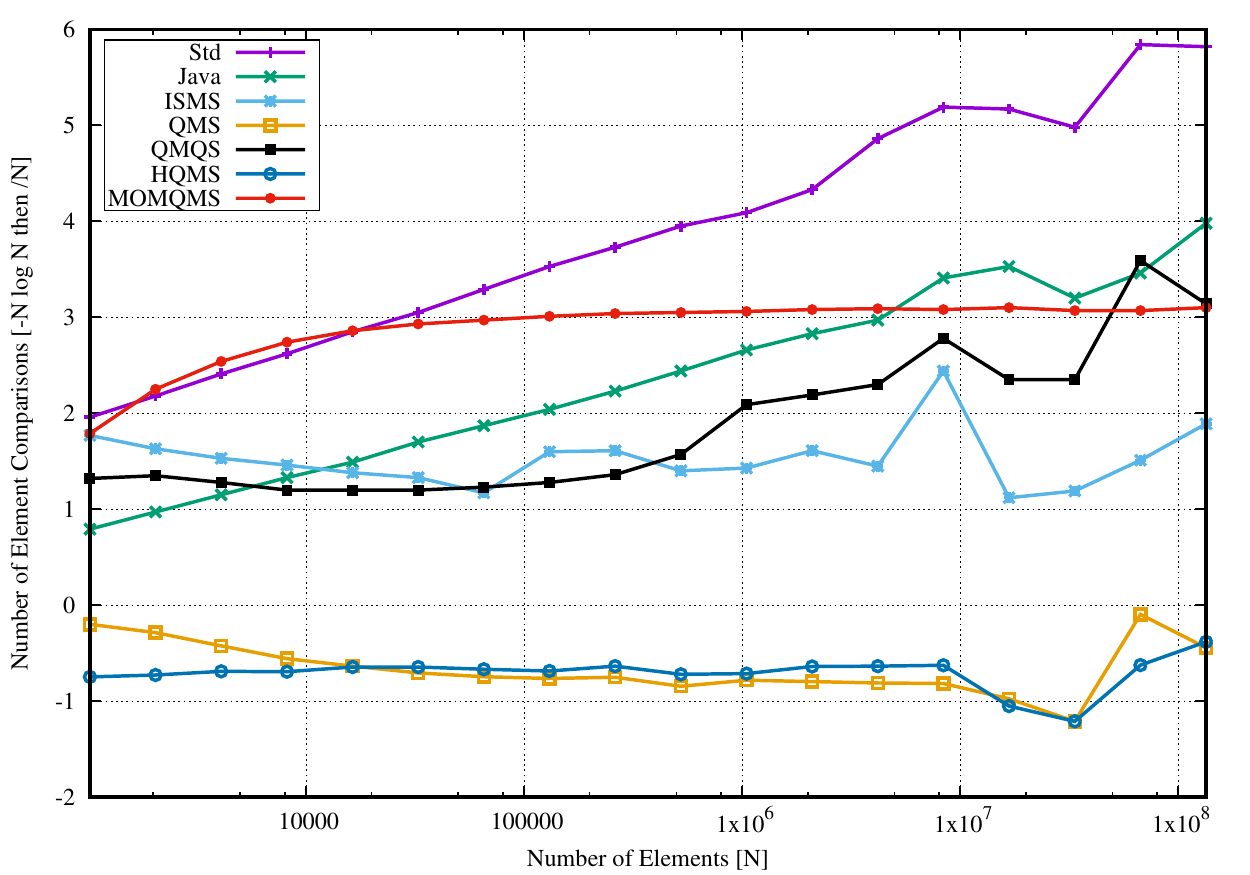}
\end{center}
\caption{Time (left) and element comparisons (right) for sorting
  random integer data.
}
\label{fig:sortrandom}
\end{figure}

\begin{figure}[ht]
\begin{center}

  \includegraphics[width=6.8cm]{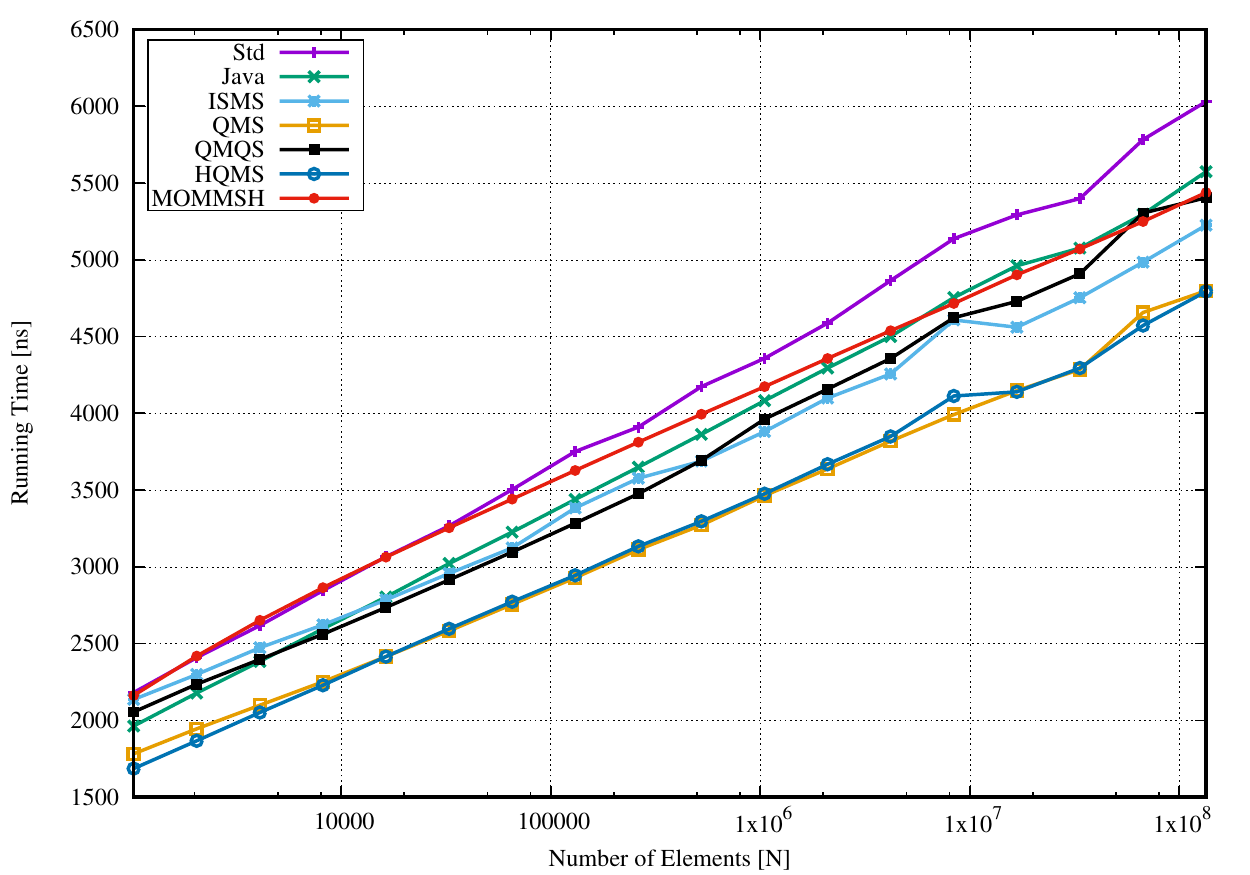}
		\includegraphics[width=6.8cm]{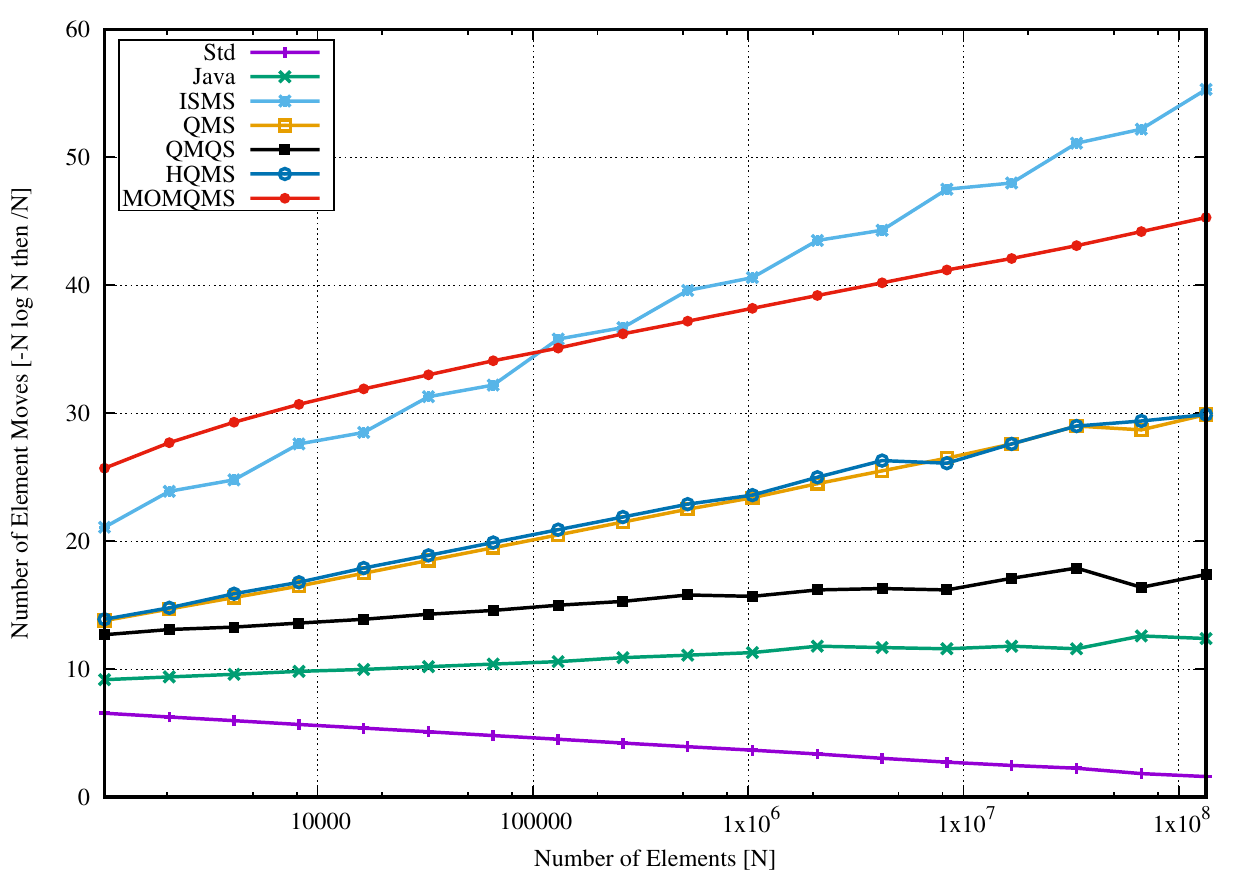}

\end{center}
\caption{Time for sorting random data with a comparator that applies the
  logarithm to the integer elements (left), and number of element
  moves (right).
}
\label{fig:sortrandomdup}
\end{figure}

\section{Conclusion}

Sorting $n$ elements is one of the most frequently studied subjects in
computer science with applications in almost all areas in which
efficient programs run.

With variants of {QuickMergesort}, we contributed sorting algorithms
which are able to run faster than {Introsort} and {DualPivotQuicksort} even for
elementary data.  Compared to Introsort, we reduced the leading term
$\alpha$ in $\alpha \cdot n \log n +\Oh(n)$ in the average number of
comparisons from $\alpha \approx 1.18$ via $1.09$ to finally reaching
1. The algorithms are simple but effective: a) Median-of-3 pivot
selection (as opposed to using a sample of $\sqrt{n}$), b)
faster sorting for smaller element sets. 
Both modifications show empirical impact and are analyzed
theoretically to provide upper bounds on the average number of
comparisons.
We discussed options to warrant a constant-factor
optimal worst-case.

In the theoretical part of our work we concentrated on average-case
analyses, as we strongly believe that this reflects realistic behavior
more closely than worst-case analyses.
With very low overhead, {QuickMergesort} has implemented in
a way that it becomes constant-factor optimal in the worst-case, too.
We chose efficient deterministic median-of-median strategies that are
also of interest for further considerations.

For future research we propose 
the integration of
QuickMergesort with multi-way merging, envisioning
to scale the algorithm
beyond main memory capacity and effective parallelizations.


\begin{thebibliography}{10}
	
	\bibitem{Alexandrescu17}
	Andrei Alexandrescu.
	\newblock Fast deterministic selection.
	\newblock In {\em 16th International Symposium on Experimental Algorithms,
		{SEA} 2017, June 21-23, 2017, London, {UK}}, pages 24:1--24:19, 2017.
	
	\bibitem{dqsanalysis1}
	Martin Aum{\"{u}}ller and Martin Dietzfelbinger.
	\newblock Optimal partitioning for dual pivot quicksort - (extended abstract).
	\newblock In {\em {ICALP}}, pages 33--44, 2013.
	
	\bibitem{dqsanalysis3}
	Martin Aum{\"{u}}ller, Martin Dietzfelbinger, and Pascal Klaue.
	\newblock How good is multi-pivot quicksort?
	\newblock {\em CoRR}, abs/1510.04676, 2015.
	
	\bibitem{BFPRT73}
	Manuel Blum, Robert~W. Floyd, Vaughan~R. Pratt, Ronald~L. Rivest, and Robert~E.
	Tarjan.
	\newblock Time bounds for selection.
	\newblock {\em Journal of Computer and System Sciences}, 7(4):448--461, 1973.
	
	\bibitem{quickheap}
	D.~Cantone and G.~Cinotti.
	\newblock Quick{H}eapsort, an efficient mix of classical sorting algorithms.
	\newblock {\em Theoretical Computer Science}, 285(1):25--42, 2002.
	
	\bibitem{CLRS09}
	Thomas~H. Cormen, Charles~E. Leiserson, Ronald~L. Rivest, and Clifford Stein.
	\newblock {\em Introduction to Algorithms}.
	\newblock The MIT Press, 3th edition, 2009.
	
	\bibitem{DiekertW13Quick}
	Volker Diekert and Armin Wei{\ss}.
	\newblock Quickheapsort: Modifications and improved analysis.
	\newblock In {\em CSR}, pages 24--35, 2013.
	
	\bibitem{Dut93}
	Ronald~D. Dutton.
	\newblock Weak-heap sort.
	\newblock {\em BIT}, 33(3):372--381, 1993.
	
	\bibitem{EWCSR}
	Stefan Edelkamp and Armin Wei{\ss}.
	\newblock {QuickXsort}: Efficient sorting with n logn - 1.399n + o(n)
	comparisons on average.
	\newblock In {\em {CSR}}, pages 139--152, 2014.
	
	\bibitem{EdelkampW16_BQS}
	Stefan Edelkamp and Armin Wei{\ss}.
	\newblock Blockquicksort: Avoiding branch mispredictions in {Q}uicksort.
	\newblock In Piotr Sankowski and Christos~D. Zaroliagis, editors, {\em 24th
		Annual European Symposium on Algorithms, {ESA} 2016, August 22-24, 2016,
		Aarhus, Denmark}, volume~57 of {\em LIPIcs}, pages 38:1--38:16. Schloss
	Dagstuhl - Leibniz-Zentrum f{\"u}r Informatik, 2016.
	
	\bibitem{ElmasryKS12}
	Amr Elmasry, Jyrki Katajainen, and Max Stenmark.
	\newblock Branch mispredictions don't affect mergesort.
	\newblock In {\em SEA}, pages 160--171, 2012.
	
	\bibitem{Flo64}
	Robert~W. Floyd.
	\newblock Algorithm 245: {T}reesort 3.
	\newblock {\em Comm. of the ACM}, 7(12):701, 1964.
	
	\bibitem{FordJ59}
	Jr. Ford, Lester~R. and Selmer~M. Johnson.
	\newblock A tournament problem.
	\newblock {\em The American Mathematical Monthly}, 66(5):pp. 387--389, 1959.
	
	\bibitem{Hoare61_find}
	C.~A.~R. Hoare.
	\newblock Algorithm 65: Find.
	\newblock {\em Commun. ACM}, 4(7):321--322, July 1961.
	
	\bibitem{Hoa62}
	Charles A.~R. Hoare.
	\newblock Quicksort.
	\newblock {\em The Computer Journal}, 5(1):10--16, 1962.
	
	\bibitem{Kaligosi}
	Kanela Kaligosi and Peter Sanders.
	\newblock How branch mispredictions affect quicksort.
	\newblock In {\em ESA}, pages 780--791, 2006.
	
	\bibitem{ultimate}
	Jyrki Katajainen.
	\newblock The ultimate heapsort.
	\newblock In {\em Computing: The Fourth Australasian Theory Symposium (CATS)},
	pages 87--96, 1998.
	
	\bibitem{KatajainenPT96}
	Jyrki Katajainen, Tomi Pasanen, and Jukka Teuhola.
	\newblock Practical in-place mergesort.
	\newblock {\em Nord. J. Comput.}, 3(1):27--40, 1996.
	
	\bibitem{Knu73}
	Donald~E. Knuth.
	\newblock {\em Sorting and Searching}, volume~3 of {\em The Art of Computer
		Programming}.
	\newblock Addison Wesley Longman, 2nd edition, 1998.
	
	\bibitem{KushagraLQM14}
	Shrinu Kushagra, Alejandro L{\'{o}}pez{-}Ortiz, Aurick Qiao, and J.~Ian Munro.
	\newblock Multi-pivot quicksort: Theory and experiments.
	\newblock In {\em {ALENEX}}, pages 47--60, 2014.
	
	\bibitem{MartinezPV10}
	Conrado Mart{\'{\i}}nez, Daniel Panario, and Alfredo Viola.
	\newblock Adaptive sampling strategies for quickselects.
	\newblock {\em {ACM} Trans. Algorithms}, 6(3):53:1--53:45, 2010.
	
	\bibitem{MartinezR01}
	Conrado Mart\'{\i}nez and Salvador Roura.
	\newblock Optimal {S}ampling {S}trategies in {Quicksort} and {Quickselect}.
	\newblock {\em SIAM J. Comput.}, 31(3):683--705, 2001.
	
	\bibitem{Mus97}
	David~R. Musser.
	\newblock Introspective sorting and selection algorithms.
	\newblock {\em Software---Practice and Experience}, 27(8):983--993, 1997.
	
	\bibitem{SandersW04}
	Peter Sanders and Sebastian Winkel.
	\newblock Super scalar sample sort.
	\newblock In {\em ESA}, pages 784--796, 2004.
	
	\bibitem{Wegener:MDR}
	Ingo Wegener.
	\newblock The worst case complexity of {McDiarmid and Reed's} variant of
	bottom-up-heap sort is less than $n \log n + 1.1n$.
	\newblock In {\em STACS}, pages 137--147, 1991.
	
	\bibitem{Weg93}
	Ingo Wegener.
	\newblock {B}ottom-up-{H}eapsort, a new variant of {H}eapsort beating, on an
	average, {Q}uicksort (if $n$ is not very small).
	\newblock {\em Theoretical Computer Science}, 118:81--98, 1993.
	
	\bibitem{dqsanalysis2}
	Sebastian Wild, Markus~E. Nebel, and Ralph Neininger.
	\newblock Average case and distributional analysis of dual-pivot quicksort.
	\newblock {\em {ACM} Transactions on Algorithms}, 11(3):22:1--22:42, 2015.
	
	\bibitem{Wil64}
	J.~W.~J. Williams.
	\newblock Algorithm 232: {H}eapsort.
	\newblock {\em Communications of the ACM}, 7(6):347--348, 1964.
	
\end{thebibliography}
\end{document}